\documentclass[letterpaper,10pt,conference]{ieeeconf}  

\IEEEoverridecommandlockouts                              

\usepackage{graphics} 
\usepackage{epsfig} 
\usepackage{amsmath} 
\usepackage{amssymb}  
 
\usepackage{amsthm}
\usepackage{cite}
\usepackage{color}
\usepackage{enumerate}
\usepackage{graphicx}
\usepackage{mathtools}
\usepackage{cancel}
\usepackage{url}

\newcommand{\naturals}{\mathbb{N}}
\newcommand{\real}{\mathbb{R}}

\newcommand{\map}[3]{#1:#2 \rightarrow #3}

\newcommand{\longthmtitle}[1]{\mbox{}{\textit{(#1):}}}

\newcommand{\setdefb}[2]{\big\{#1 \; | \; #2\big\}}
\newcommand{\setdefB}[2]{\Big\{#1 \; | \; #2\Big\}}
\newcommand*{\SetSuchThat}[1][]{} 
\newcommand*{\MvertSets}{%
    \renewcommand*\SetSuchThat[1][]{%
        \mathclose{}%
        \nonscript\;##1\vert\penalty\relpenalty\nonscript\;%
        \mathopen{}%
    }%
}
\MvertSets 
\DeclarePairedDelimiterX \Set [2] {\lbrace}{\rbrace}
    {\,#1\SetSuchThat[\delimsize]#2\,}

\newcommand{\Cc}{\mathcal{C}}

\newcommand{\Kc}{\mathcal{K}}

\newcommand{\defeq}{\triangleq}

\newtheorem{theorem}{Theorem}
\newtheorem{assumption}[theorem]{Assumption}

\newtheorem{remark}[theorem]{Remark}

\newtheorem{proposition}[theorem]{Proposition}  

\theoremstyle{definition}

\newcommand{\nom}{{\operatorname{nom}}}

\newcommand{\on}{{\operatorname{on}}}
\newcommand{\off}{{\operatorname{off}}}

\newcommand{\predict}{{\operatorname{p}}}
\newcommand{\Lie}{\mathcal{L}}

\renewcommand{\baselinestretch}{1}
\allowdisplaybreaks

\DeclareMathOperator*{\argmin}{argmin}

\title{\LARGE \bf Intermittent Safety Filters for Event-Triggered Safety Maneuvers with Application to Satellite Orbit Transfers}
\author{Pio Ong and  Aaron D. Ames
\thanks{This research is supported in part by the National Science Foundation (CPS Award \#1932091) and Raytheon.}
  \thanks{Pio Ong and Aaron D. Ames are with the Department of Mechanical and Civil Engineering, California Institute of Technology, Pasadena, CA 91125, USA. {\tt\small \{pioong,ames\}@caltech.edu}}%
}
\begin{document}

\maketitle
\begin{abstract}
    In balancing safety with the nominal control objectives, e.g., stabilization, it is desirable to reduce the time period when safety filters are in effect. Inspired by traditional spacecraft maneuvers, and with the ultimate goal of reducing the duration when safety is of concern, this paper proposes an event-triggered control framework with switching state-based triggers. Our first trigger in the scheme monitors safety constraints encoded by barrier functions, and thereby ensures safety without the need to alter the nominal controller---and when the boundary of the safety constraint is approached, the controller drives the system to the region where control actions are not needed. The second trigger condition determines if the safety constraint has improved enough for the success of the first trigger. 
    We begin by motivating this framework for impulsive control systems, e.g., a satellite orbiting an asteroid. We then expand the approach to  more general nonlinear system through the use of safety filtered controllers. Simulation results demonstrating satellite orbital maneuvers illustrate the utility of the proposed event-triggered framework.
\end{abstract}
\section{Introduction}

The idea of filtering a nominal control action to satisfy safety constraints---that is \emph{safety filtering} \cite{TG-AS-JR-LC-EF-ADA:18,KH-MM-MA-SC-EF:21}---has proven to be a powerful tool in controller synthesis. This technique facilitates the controller design process as we can decouple different control objectives, e.g., stability and safety.  An unfortunate consequence of applying safety filters is the possibility that the deviation from the nominal controller will affect the success of the nominal objective, e.g., stability. Addressing this potential conflict involves bounding the deviation from, and the effect on, the nominal controller---a difficult task in nonlinear system. 

This paper presents a new approach to safety filtering that is \emph{intermittent} in nature, with the result being event-triggered safety maneuvers.   
To this end, we take inspiration from safety maintenance of a satellite via orbit transfers \cite{AHDR-CD-JRF:12}. Satellites spend most of their time during a given mission with the nominal objectives, and only apply corrective maneuvers intermittently to ensure safety, e.g., in response to unmodelled environmental dynamics. Because there are periods where there is no deviation from the nominal controller, guarantees can be made if enough time is spent in these safe operating regions. The goal of this paper is to develop an event-triggered framework that imitates the behavior of satellites that has proven useful in practice. In particular, the switching between safety objective in an ``on demand'' fashion, and the application of safety filters in an intermittent fashion.
The end result will be the introduction of \textit{event-triggered intermittent safety filtering}. 
The hope is that this will lay the necessary groundwork for an alternative approach 
enforcing safety guarantees while minimally modifying the nominal performance objectives.  

\begin{figure}[t]
\label{fig:fit}
\centering
\includegraphics[width=\linewidth,height=\textheight,keepaspectratio]{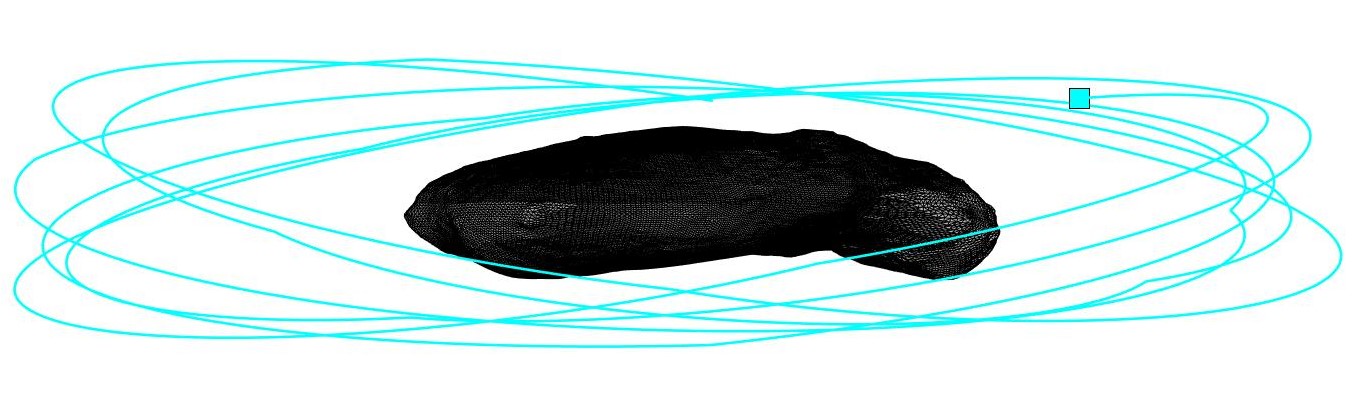}
\caption{Satellite safe orbit trajectory around 25143 Itokawa utilizing event-triggered safety maneuvers.}
\end{figure}

\subsection*{Literature Review} 
 
Our paper relies on two main bodies of literature: safety-critical control via certificates, e.g., control barrier functions, and event-triggered control. 

Safety-critical controls aims to provide a framework for formally guaranteeing the safety of nonlinear control systems---with safety typically framed as avoiding undesirable states, or equivalently always staying in a set of desirable states. Barrier certificate~\cite{SP-AJ:04} describes a safety of a set utilizing a scalar function, wherein it is simpler to analyze corresponding violations. In this context, barrier certificates are related to ideas underlying nonovershooting control that restricts outputs~\cite{MK-MB:07} or Lyapunov function evolving within a specified bound~\cite{PO-AB-TH-LK-PS:06}.  Nevertheless, the attractiveness of barrier certificate is that it isolates safety problem from the design of the controller for other objectives, thereby transforming the problem into a more tractable form. 

Control barrier functions as first introduced~\cite{PW-FA:07} mirror control Lyapunov functions by studying barrier certificates associated with control systems. Yet the result was overly conservative, and as a result the modern form of control barrier functions (CBFs) where introduced in~\cite{ADA-JWG-PT:14,XX-PT-JWG-ADA:15,ADA-XX-JWG-PT:17,ADA-SC-ME-GN-KS-PT:19}; these are necessary and sufficient for forward set invariance thereby generalizing Nagumo theorem~\cite{FB-SM:07} to control systems. The non-conservative nature of CBFs lead to the idea of a safety filter \cite{TG-AS-JR-LC-EF-ADA:18}---and optimization-based controller that minimally modifies a nominal controller to ensure safety. This paper seeks to expand the notion of a safety filter by shortening the filtering period, thereby allowing the nominal controller more freedom. The idea is made possible by event-trigger control.

Event-triggered control~\cite{PT:07,WPMHH-KHJ-PT:12,LH-CF-HO-AS-EF-JPR-SIN:17} synthesizes property preserving discrete implementations of continuous controllers. The seminal work~\cite{PT:07} studies sample-and-hold system and proposes an aperiodic sampling scheme that is based on system states rather than fixed time. Event-triggered control shows great improvement in the average sampling time. Event-triggered control implemented in the context of sample-and-hold can be formalized as an impulsive system~\cite{MCFD-WPMH:12}, and in general as a hybrid system \cite{JC-PC-RGS:20,JC-PC-RGS:17}. This hints at the possibility of transferring the trigger ideas and applying them to different type of hybrid systems. 

In this paper, we study event-triggered control of impulsive systems in the context of safety; this can, for example, be used to describe satellite systems. There are works on event-triggered impulsive control, but they focus on stabilization~\cite{BL-DNL-CXD:14, BL-DH-ZS:18,XL-DP-JC:20, KZ-BG:21} which is not applicable to safety because of difference in two objectives, cf. ~\cite{AJT-PO-JC-AA:21-csl}. Additionally, there are works that deal with safety~\cite{GY-CB-RT:19,AJT-PO-JC-AA:21-csl,WX-CB-CGC:21} but not for impulsive control systems.   Finally, this paper also studies event-triggered control in the context of intermittent control. Our previous work~\cite{PO-GB-ADA:22-cdc} also considers intermittent control, but in the context of sample-and-hold whereas this paper considers the frequency of safety filtering.

\subsection*{Statement of Contribution}

This paper investigates the idea of lengthening the time period during which a safety-critical controller does not need to actively spend control effort in order to satisfy safety constraints.  To this end, we use event-triggered control with barrier functions to monitor safety and opportunistically determine when to start applying control for the purpose of satisfying safety objectives. The first contribution of this paper is the formalization of an event-triggered implementation of safety-critical controllers for impulsive control systems. We provide a trigger design that guarantees safety, along with the sufficient conditions for establishing the minimum inter-event time, which rules out the possibility of Zeno behavior. The second contribution is another trigger design for impulsive system that is built on the first trigger design, with the goal of imitating the key elements of a satellite maneuver. To accomplish this, we add an additional trigger condition that monitors the increase of the inter-event time of the subsequent application of control. 

To demonstrate the effectiveness of the two trigger designs presented in this paper, we consider the problem of maintaining a satellite within an orbital radius range---the \emph{satellite orbit safety problem}. For the second trigger design, we demonstrate its similarity to traditional orbit transfers. This motivates the third contribution of the paper: an event-triggered intermittent safety filter framework, which is developed as an extension of our first trigger design in the context intermittent control systems. This framework takes inspiration from satellite maneuvers, i.e., our design is based on (1) maneuvering to safer states in order to avoid the need to filter the nominal controller, and (2) monitoring the barrier function to determine when safety is critical. Our framework actively allows a period for where nominal controller may be applied without a safety filter, and is thus able to make progress towards nominal objectives. 

\vspace{0.1cm}
\subsubsection*{Notation}
We denote with $\naturals$ and $\real$ the set of all natural and real numbers respectively. For a vector $x\in\real$, $\|x\|$ is its Euclidean norm. A  function $\map{\alpha}{\real}{\real}$ is of class-$\Kc$ if $\alpha(0)=0$, and $\alpha$ is strictly increasing. Let $t\mapsto x(t)$ be a solution to a dynamical system for an initial condition $x_0$, a set $\Cc$ is forward invariant and safe if $x(t)\in\Cc$ for all time whenever $x_0\in\Cc$.

\section{Preliminaries}
We begin by providing some background on the safety concept and the common practice of using safety filter, in order to motivate the problem we consider in this paper. In addition, we provide the background on event-triggered control, which we believe to be the key tool for solving our problem. Here, we consider the nonlinear system
\begin{equation} \label{sys:nonlinear}
    \dot x = f(x,u)+ d
\end{equation}
where $x\in\real^n$ is the state, $u\in \real^m$ is the control input, and $d\in\real^n$ is the disturbance to the system. We assume throughout the paper that the disturbance is bounded, i.e., $\|d\|\leq \bar d$.

\subsection{Safety Formulations and Safety Filters}
For safety problems, we are interested in ensuring that the states along system trajectories are not undesirable states. To address these problems, one approach is to define the safe set~$\Cc$, consisting of only ``safe" states via a \textit{barrier function}~$\map{h}{\real^n}{\real}$ such that:
\begin{subequations}\label{eq:safeset}
\begin{align}
\Cc &= \setdefb{x \in \real^n}{h(x) \geq 0}, \\
\partial \Cc &= \setdefb{x \in \real^n}{h(x) = 0}, \\
\text{Int}(\Cc) &= \setdefb{x \in \real^n}{h(x) > 0}.
\end{align}
\end{subequations} 
With the barrier function, the safety problems involve finding the input signal $t\mapsto u(t)$ that ensures $h$ remains positive at all time so that the safe set $\Cc$ forward invariant, and the trajectories do not reach the undesirable states. To this end, we can use a state-feedback control $u=k(x)$ with a controller $\map{k}{\real^n}{\real^m}$ satisfying the following \textit{barrier condition}:
\begin{equation}\label{eq:barr_cond}
    \underbrace{\left. \frac{\partial h}{\partial x} \right|_{x} f(x,k(x))}_{\defeq \Lie_fh(x,k(x))} - \left\|\left. \frac{\partial h}{\partial x} \right|_{x}\right\| \bar d\geq -\alpha(h(x)).
\end{equation}
for some class-$\Kc$ function $\alpha$. The condition above is a conservative way to keep the function $h$ positive. Not only does it require $h$ to not decrease whenever it is zero, but it also requires that $h$ does not decrease too quickly as it gets closer to zero. Nevertheless, it provides many benefits in many applications, one of which is the ability to monitor safety. We will discuss this point later in the paper.

More often than not, safety is not the only objective for the control system. In such case, a nominal controller $\map{k_\nom}{\real^n}{\real^m}$ is first designed to meet other objectives. Then a safety filter is used to enforce the barrier condition by defining the controller using an optimization:
\begin{align}
k(x) = &\argmin_{u\in\real^m} \|u-k_\nom(x)\|^2\label{eq:filter}\\
&\quad\text{s.t.}~~ \Lie_fh(x,u)- \left\|\left. \frac{\partial h}{\partial x} \right|_{x}\right\| \bar d\geq -\alpha(h(x)).\nonumber
\end{align}
The idea is that we prioritize safety above other objectives, so we adjust the nominal controller so that barrier condition is always maintained. The optimization assures that the resulting controller deviates from the nominal controller as little as possible (minimal Euclidean distance in this case). 

Nevertheless, the deviation from the nominal controller can pose an issue in using safety filters because the filtered controller may no longer satisfy the original objective. This motivates the problem we seek to address.

Consider the state-feedback control $u=k(x)$ of the nonlinear system~\eqref{sys:nonlinear}. We can separate the periods when the constraint of the filter~\eqref{eq:filter} is active and inactive and turn it into an intermittent nonlinear system as:
\begin{equation}\label{sys:intermittent}
    \dot x =\begin{cases} f(x,k(x))+d, &t\in[t^\on_i,t^\off_i)\\
    f(x,k_\nom(x))+d, &t\in[t^\off_i,t^\on_{i+1}).
    \end{cases}
\end{equation}
Here, the time at which the the filter is on and off is automatically determined by whether or not the constraint is active. Notice that it is possible for the off period to be nonexistent if the constraint never becomes inactive. In this work, we identify a method to assure the existence of the off period and we use event-triggered control to lengthen the off period for as long as possible.


\subsection{Event-Triggered Control}
Under the event-triggered control framework, the control is applied according to a \textit{trigger condition}, instead of having a scheduled time for control. Such condition is usually based on the states of the system, so the control is applied only when necessary. Trigger designs are often written as:
\begin{equation}\label{eq:trigger}
t_{j+1} = \setdefb{t\geq t_j}{\Xi(x(t))\leq 0}
\end{equation}
with a trigger condition $\map{\Xi}{\real^n}{\real}$. The trigger enforces $\Xi(x(t))> 0$ for the duration $[t_j,t_{j+1})$ because the control is applied otherwise. 
Event-triggered control is often used as a tool to reduce the frequency of control adjustments because they are done only when necessary. Event-triggered control are often studied in the context of sample-and-hold.

\subsection{Sample-and-hold systems}
Such systems appear naturally as control systems are implemented on a digital platform. Here, we model sample-and-hold control systems as impulsive control systems. Consider a nonlinear system:
$$
\dot y = F(y,v)
$$
where $y\in \real^{n_y}$ is the system states and $v\in \real^{m_v}$ is the control input for the system dynamics $\map{F}{\real^{n_y}\times \real^{m_v}}{\real^{n_y}}$. Sample-and-hold systems samples the control input signal at different time instances $t_j$, then holds the value constant, $v(t)=v(t_j)$, for time $t\in[t_j,t_{j+1})$. We can rewrite the system with the control input $v$ as part of the state:
\begin{equation*}
\begin{bmatrix} \dot y \\ \dot v \end{bmatrix} = \begin{bmatrix} f(y,v)\\ 0 \end{bmatrix},~
\begin{bmatrix} y \\ v \end{bmatrix}^+ = \begin{bmatrix}  y \\ v + \Delta v \end{bmatrix}.
\end{equation*}
Defining the state $x = \begin{bmatrix} y \\  v \end{bmatrix}$ and the control input $u=\Delta v = v(t_{j+1})-v(t_j)$, we have transformed the sample-and-hold control system as an impulsive control system~\eqref{sys:flow}-\eqref{sys:jump}. While there are many impulsive systems in the real world, like the satellite system we consider in this paper, impulsive systems being a generalization of sample-and-hold control systems makes them even more interesting to study, especially under the context of event-triggered control.

\section{Event-Triggered Impulsive Safety}
\label{sec:ET-impulse}
We begin the exposition with the consideration of impulsive control systems. Such systems can be modelled as hybrid systems with flow dynamics $\map{F}{\real^n}{\real^n}$ and a jump map $\map{G}{\real^n\times\real^m}{\real^n}$ as:
\begin{subequations}\label{sys:impulse}
    \begin{align}
        \dot x &= F(x) + d, \label{sys:flow}\\
        x^+ &= G(x,v) \label{sys:jump}
    \end{align}
\end{subequations}
where $x\in\real^n$, $v\in\real^m$ and $d\in\real^n$ are the system state, the control input, and the flow disturbance. One example for this type of system is the satellite system, which we will give more details later.

We are interested in safety problems for the impulsive systems~\eqref{sys:impulse}. The main difference between safety problems and stabilzation problems is the fact that we must satisfy safety criteria at all points along the trajectory. Notice however that the control input does not appear in the flow~\eqref{sys:flow}. If we consider the flow dynamics~$F$ to be the result of closing the loop for nonlinear feedback system $F(x) = f(x,k_\nom(x))$ with the nominal controller, then the flow is the duration for which we cannot use safety filters because we cannot make any adjustment. At the same time, our problem of avoiding safety filters get simplified in impulsive control systems because the effect of control inputs are instantaneous. In this case, we are effectively ignoring the filtering duration, and we can focus on understanding what condition we may want as the filtering duration ends.

\subsection{Safeguarding Impulsive Controller}
For a safeguarding impulsive controller~$\map{K}{\real^n}{\real^m}$, we consider two following objectives. First, we require that each jump results in a state that remains in the safe set:
\begin{equation}\label{eq:jump_positive}
    h(G(x,K(x))) \geq 0.
\end{equation}
This requirement is straightforward and is relatively easier to meet. In addition, we need the states along the trajectory during the flow to meet the barrier condition: 
\begin{equation}\label{eq:flow_barrier}
    \Lie_Fh(x) - \left\|\left. \frac{\partial h}{\partial x} \right|_{x}\right\| \bar d \geq -\alpha\big(h(x)\big),
\end{equation}
which is more problematic to satisfy. We have to rely on our control input at each jump to guarantee safety of the ensuing trajectory during the flow, up until the next jump occurrence. To address this second requirement, our approach is to design the controller~$K$ so that:
\begin{multline}\label{eq:control_jump_cond}
\Lie_Fh(G(x,K(x))) - \left\|\left. \frac{\partial h}{\partial x} \right|_{G(x,K(x))}\right\| \bar d \\ \geq -\alpha\big(h(G(x,K(x)))\big) + c.
\end{multline}
with some positive constant $c$. The idea is to use the constant $c$ to provide a buffer so that some time has to elapse (due to continuity) after each jump, before the barrier condition~\eqref{eq:flow_barrier} gets violated.

\subsection{Event-Triggered Safeguarding Impulsive Control}
We aim at reducing the frequency of the controls, and a reasonable approach is to employ event-triggered control to prescribe when to apply controls. To this end, our main trigger condition is based on monitoring the barrier condition~\eqref{eq:flow_barrier}. Denoting $t_i$ as the last instance when the jump occur, we determine the jump time iteratively with:
\begin{subequations}
    \label{trigger:greedy}
\begin{align}
    t_{i+1} &= \min\setdefB{t\geq t_{i}}{\Xi(x(t))\leq 0},\\
    \Xi(x) &= \Lie_Fh(x) - \left\|\left. \frac{\partial h}{\partial x} \right|_{x}\right\| \bar d  +\alpha\big(h(x)\big).\label{eq:greedy_cond}
\end{align}
\end{subequations}
The trigger above makes sure that controls only when necessary, i.e., when the barrier condition is violated. This strategy is a greedy approach to maximizing the times between jumps and reducing how often the jumps occur. 

The main concern when relying on event-triggered control is the possibility of Zeno behavior. That is, because the elapsed time between $t_i$  and $t_{i+1}$ are not uniform for all $i\in\naturals$, it become possible that the sequence $\{t_i\}_{i\in\naturals}$ can converge to a constant value rather than infinity. This would mean there can infinite numbers of jumps in finite time, making it impractical to implement in practice.

The common approach of ruling out Zeno behavior is to establish a minimum inter-event time (MIET), i.e., a common positive lower bound $\tau \leq t_{i+1}-t_i$ for all $i\in\naturals$. The task of establishing a MIET is often difficult, especially when the system is subjected to an unknown disturbance $d$. In this paper, we endow our controller with condition~\eqref{eq:control_jump_cond}. This means that whenever a jump occurs, the trigger condition~$\Xi$ has a value greater than $c$. Thus, if the rate at which trigger condition can decrease is bounded by a constant, a MIET can be derived. A set of assumptions we can make to achieve this is given in the following result.

\begin{proposition}\longthmtitle{Event-Triggered Safety for Impulsive Control Systems}
\label{prop:ET-impulse}
Consider the impulsive control system~\eqref{sys:impulse} with jump time $\{t_i\}_{i\in\naturals}$ determined iteratively by the trigger design~\eqref{trigger:greedy} and the corresponding control satisfying~\eqref{eq:control_jump_cond}. Assume the followings:
\begin{enumerate}[(i)]
    \item the flow dynamics~$F$ is bounded on $\Cc$;
    \item the trigger condition~$\Xi$ is Lipschitz and continuously differentiable on $\Cc$.
\end{enumerate}
Then there exists a MIET~$\tau \leq t_{i+1}-t_i$. As a consequence, $x(t)\in\Cc$ for all time if $x_0\in\Cc$. That is, the set $\Cc$ is safe. 
\end{proposition}
\begin{proof}
Let $B\geq \|F(x)\|$ denote the bound of the flow dynamics and $L_\Xi\geq \left\|\left.\frac{\partial \Xi}{\partial x}\right|_x\right\|$ denote the Lipschitz constant of the trigger condition. We can estimate the lower bound of the value of the trigger condition as follow:
\begin{align*}
\Xi(t) &= \Xi(x(t_i)) + \int_{t_i}^t\left.\frac{\partial \Xi}{\partial x}\right|_{x(t)} (F(x(t))+d) dt \\
&\geq c - \int_{t_i}^t L_\Xi (B+\bar d) dt \\
&= c - L_\Xi (B+\bar d) (t-t_i).
\end{align*}
Hence, it is only possible for $\Xi(t) \leq 0$ when $t \geq c/( L_\Xi (B+\bar d)) +t_i = \tau+t_i$. Therefore, $t_{i+1}-t_i \geq  \tau$ and the possibility of Zeno behavior is ruled out.

Without Zeno behavior, system trajectories are defined at all time. Now note that \eqref{eq:jump_positive} is by design of the controller, so $x(t_i)\in \Cc$ for all $i\in\naturals$. Then, we may conclude $x(t)\in \Cc$ because the barrier condition \eqref{eq:flow_barrier} is satisfied at all time due to the trigger~\eqref{trigger:greedy}. This concludes the proof.
\end{proof}
Proposition~\ref{prop:ET-impulse} provides an event-triggered implementation solution of an impulsive safeguarding controller. We have made two regularity assumptions in order to sufficiently establish the MIET. Even though MIET is not required for deducing safety (unlike the case of stablization), it is important that our trigger scheme is practical. Next, we demonstrate the effectiveness of our trigger design through satellite safety problem.

\subsection{Application to Satellite Systems}
Satellites orbiting around a central body can be described by Newton's gravitation model. Particularly, denoting the position and velocity vectors of a satellite with $\vec{r}\in \real^3$ and $\vec{v}\in\real^3$, the satellite is subjected to the dynamics from the gravity field:
$$
\frac{d}{dt}\begin{bmatrix}\vec{r}\\ \vec{v}\end{bmatrix} = \begin{bmatrix}\vec{v}\\- \frac{\mu}{r^3} \vec{r}\end{bmatrix}+d
$$
where $\mu$ is the gravitational parameter of the central body, $r$ is the shorthand notation for $\|\vec{r}\|$, and $d\in\real^3$ is the disturbance to the dynamics such as the higher order gravity field not considered in the Newton's model. 


The satellite is controlled by firing thrusters to apply a change in velocity~$\Delta \vec{v}$. This change is assumed instantaneous:
$$
\begin{bmatrix}\vec{r}\\ \vec{v}\end{bmatrix}^+ = \begin{bmatrix}\vec{r}\\ \vec{v}\end{bmatrix}+ \begin{bmatrix}0\\ \Delta \vec{v}\end{bmatrix}.
$$
In reality, the thruster firings are not impulses; instead, they last for a few seconds. However, this timescale is much smaller than the time elapsed between each firing, which is in the scale of tens to hundreds of hours. Thus, the impulsive approximation is often used in orbital mechanics for satellite maneuvers, and we adopt this model.

We simulate an application of the impulsive trigger design~\eqref{trigger:greedy} to the satellite system. The satellite is orbiting around an asteroid, 25143 Itokawa, and the disturbance~$d$ comes from the unmodelled higher order gravity field, which is due to the asteroid not being a perfect sphere.  In term of safety, we are interested in maintaining the satellite within a range of desirable orbital distance. In our example, we want the satellite to orbit in the range $1.6R\leq r \leq 2.4R$ where $R$ is the mean radius of the central body. The barrier function we will use is:
$$
h(\vec{r},\vec{v}) = (0.4R)^2-(r-2R)^2.
$$
We design a safeguarding controller based on our understanding of orbital mechanics~\cite{BMW:71} of Newton's gravity model. For the reasons of space and the background needed to explain it, we omit details and reasonings in this paper. The overall explanation is that we apply impulses to inject the satellite in an orbit (without changing planes) towards the orbital radius $r_\text{target}(r)= 2R+0.5(r-2R)$ at peri/apoapsis. We place the satellite at the true anamoly that varies linearly $-\pi$ to $-\pi/2$  (or 0 to $\pi/2$) depending on the current orbital distance. 

We simulate the satellite orbiting the asteroid for 2400 hours. Fig.~\ref{fig:greedy} shows the result of our simulation for the first 150 hours of satellite orbit. We report that the satellite remains within the specify safe range, and the barrier condition does not get violated. The trigger sporadically occurs for a total of 267 times across the 2400 hours, which is approximately 1 trigger every 9 hours. 

\begin{figure}[thpb]
\label{fig:greedy}
\centering
\includegraphics[width=\linewidth,height=\textheight,keepaspectratio]{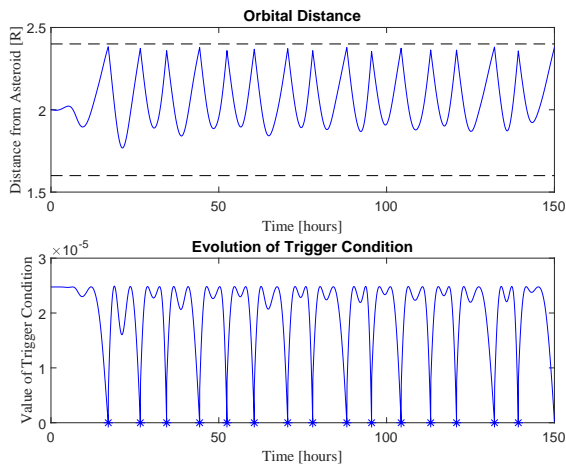}
\caption{[Top] Distance to the asteroid over 150 hours of satellite orbit. The plot shows the distance remains within the safe range at all time. [Bottom] Trigger condition over 150 hours of satellite orbit. }
\end{figure}


\section{Impulsive Safety Maneuvers}
\label{sec:maneuver}
In this section, we investigate the concept of safety maneuver as a way to improve the performance of event-triggered control using barrier condition. Our idea is inspired by how satellite maneuver via orbit transfers.

\subsection{Inspiration from Orbit Transfers}
Here, we discuss the common practice to guarantee safety for satellites in term of maintaining within a desirable range of orbital radius. This will serve as a reference point and a motivation to our approach. 

Typically, desired orbits are designed for the satellites so that they would be safe at all points along their orbits. Then to mitigate the effects from disturbances, satellite maneuvers are performed periodically, e.g. once a day, to reset the satellites back to the desired orbit. Simulations on the satellites from different initial positions on the desired orbit, i.e, Monte Carlo analyses, are performed in order to study the deviations from the desired orbits under disturbances and to ensure all safety criteria. From this analyses, an acceptable frequency of maneuvers can be found.

A satellite maneuver consists of two different impulses. The first impulse aims at repositioning the satellite to a point along the desired orbit. Once reached, a second impulse is applied to adjust the velocity to insert the satellite into the desired orbit. The maneuver is usually performed relatively quickly, i.e., within less than an hour for asteroid orbits. We believe there are benefits to this traditional approach, and we seek develop our version of \textit{safety maneuver} using event-triggered control.


\subsection{Inter-event Time Improvement}
We believe the success of the strategy relies on the improvement in inter-event time at the desired orbit. For example, in our problem of keeping a satellite within a certain range, a typical desired orbit is a circular orbit at the midpoint of the range (r=2R). Fig.~\ref{fig:fit} shows the median time at different radii, for trigger to occur after the control is applied. It should be noted that the expected inter-event time towards the center of the safe set is superior to those close to the boundary. It is cost-effective to apply two control instances (one to reposition the satellite) to enjoy the longer inter-event time.

\begin{figure}[thpb]
\label{fig:fit}
\centering
\includegraphics[width=\linewidth,height=\textheight,keepaspectratio]{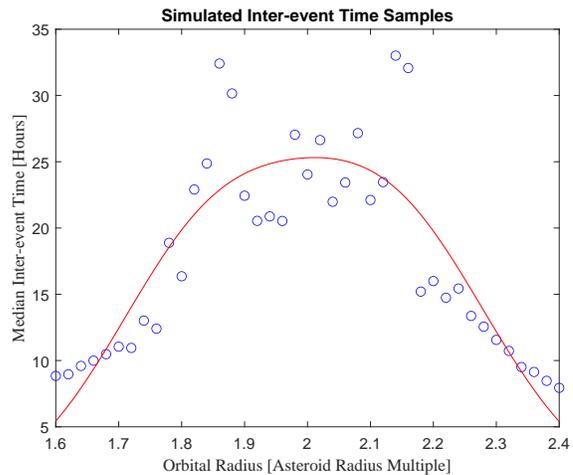}
\caption{We randomize 100 satellite positions at each orbital radius. For each position, we apply the impulsive safeguarding controller and collect the time it takes for the trigger to occur. The plot shows the median time values across the 100 samples at each radius.}
\end{figure}

The idea of increase in inter-event time can be abstracted in the context of barrier function. We define a function $\map{\tau_\predict}{\real}{\real}$ to be such that $\tau_\predict(h(x))$ is the expected inter-event time when the control is applied at the value of barrier function $h(x)$. If we have such a function, we can find its rate of change via:
$$
\dot \tau_\predict (x) = \left.\frac{d\tau_\predict}{dh}\right|_{h(x)}\Lie_Fh(x),
$$
which we can monitor along trajectory. In the satellite example, we are essentially assume $\frac{d\tau_\predict}{dh}$ is always positive, and thus, the maximum value of $h$ translates to the longest inter-event time. 

We assume the function~$\tau_p$ is obtained via data collection. Much like what we have done in Fig~\ref{fig:fit}, a likely scenario is that inter-event time data are collected for each value of $h$, and then, a curve fitting technique is performed along those data points. In the perfect scenario, the full knowledge of $\tau_p$ is preferably known as a function of state~$x$, rather than the value of barrier function~$h$. However, our approach uses the barrier function $h$ as a proxy to reduce the sampling dimension and the number of samples needed. We believe this is a good strategy because $h$ can affect inter-event time in a significant way. Referring to the trigger condition~\eqref{eq:greedy_cond}, our underlying logic is that higher $h$ will increase the value of the trigger condition $\Xi$, and thus, it would take longer time for it to reach zero. 

Indeed, our logic is not perfectly sound. The inter-event time does not depend only on the value of the barrier function. There are many variables involved such as how fast the trigger condition changes, how~$\Lie_F(x)$ changes with respect to~$h$, and how much~$h$ affects the overall value of the trigger condition via~$\alpha$. Nevertheless, the collected data will reflect that, and the function~$\tau_\predict$ will simply not be useful. However, if everything aligns, then we can obtain a function~$\tau_\predict$ that we can exploit.

\subsection{Event-triggered Impulsive Safety Maneuver}
Our safety maneuver is based on monitoring the barrier condition and the expected inter-event time. Each maneuver consists of two impulses. We note that, just like the satellite maneuvers, these two impulses do not need to be sampled from the same safeguarding controller. However, for simplicity, we will consider only one common safeguarding controller.

In order to maintain safety, both impulses rely on the trigger design~\eqref{trigger:greedy}. Let $t_i$ be the last control application, the time of first impulse~$t_{i+1}$ is determined solely according to the trigger design~\eqref{trigger:greedy}. On the other hand, the second impulse is designed to be less myopic. The trigger will not wait until the violation of safety to maximize its immediate inter-event time. Instead, we allow the second impulse instance~$t_{i+2}$ to occur prematurely if continuing on will reduce expected average inter-event time. More precisely, we consider the average between the current inter-event time and the expected inter-event time after an impulse if one were to be applied:
$$
\big((t-t_k)+\tau_\predict(h(x))\big)/2.
$$
To optimize the above quantity, we simply monitor the trigger condition:
\begin{equation}
    \Xi_\tau(x) = (1+\dot\tau_\predict)/2.
\end{equation}
The trigger makes sure that we reach the local maximum point before we apply controls. However, the optimized average may be below the current expected inter-event time $\tau_\predict(h(x(t_i)))$. To avoid this, the trigger condition~$\Xi_\tau$ will only be considered after $t_{i+1} +\tau_\predict(h(x(t_i)))$. Mathematically, our second trigger is given by:
\begin{align}
    t_{i+2} &= \min\{t_{i+2}^\text{safe}, t_{i+2}^\tau\},\label{trigger:int_time}\\
    t_{i+2}^\text{safe} &= \min\setdefb{t\geq t_{i+1}}{\Xi(x(t))\leq 0}, \nonumber\\
    t_{i+2}^\tau &=\min\setdefb{t\geq t_{i+1}+\tau_\predict(h(x(t_{i+1})))}{\Xi_\tau(x(t))\leq 0}.\nonumber
\end{align}
Because both triggers contain the monitoring of barrier condition, we can conclude the same safety guarantee. We state this formally as follows.
\begin{proposition}\longthmtitle{Event-triggered Impulsive Safety Maneuvers}
    Consider the impulsive control system~\eqref{sys:impulse} with jump time $\{t_i\}_{i\in\naturals}$ determined iteratively by switching trigger designs~\eqref{trigger:greedy} and~\eqref{trigger:int_time}. Under the same set of assumptions as in Proposition~\ref{prop:ET-impulse}, $x(t)\in\Cc$ for all time if $x_0\in\Cc$. That is, the set $\Cc$ is safe.~\hfill$\blacksquare$ 
\end{proposition}

We have proposed an alternative event-triggered scheme for maintaining safety in an impulsive control system. In the scheme, the trigger conditions switch between being greedy in maximizing immediate inter-event time and predicting one step ahead in term of maximizing the inter-event time. We note that the trigger scheme with only trigger design~\eqref{trigger:int_time} can also work in term of safety guarantee, but it is unclear whether doing so will improve in term of inter-event times.

\begin{remark}\longthmtitle{Inter-event Time Heuristic}
{\rm 
    For our result, we can claim that the average inter-event time of two consecutive flow periods would be higher than without maneuver. However, this does not guarantee the an overall increase in average inter-event time. This is because each trigger design creates a different trajectory, so it is impossible to make a guarantee on the overall average inter-event time.
~\hfill$\bullet$
}
\end{remark}

\begin{remark}\longthmtitle{Maneuver Behavior with Safety Promoting Controller Codesign}
{\rm
    Our trigger scheme takes an opportunistic approach in extending the inter-event time. To fully imitate safety maneuver behavior, the first impulsive control must actively 
    try to drive the system state to a position where the inter-event time may increase, e.g., safer location with higher value of barrier function~$h$. This would involve a codesign of the controller---designing the controller with the expectation of using our trigger scheme. For our satellite example, we assure that each impulse would promote safety in order to take full advantage of the trigger design. In our following simulation, we demonstrate the success in imitating maneuver behavior.~\hfill$\bullet$
}
\end{remark}

\subsection{Simulation Result}
We simulate our impulsive safety maneuver trigger scheme for the satellite safety problem explained earlier in the paper. In addition, we use the inter-event time data collected shown in Fig.~\ref{fig:fit} to fit a curve in order to estimate the function $\tau_\predict$. Fig.~\ref{fig:int_time} shows the results for the first 150 hours of the simulated orbital time. Safety is maintained as expected. In addition, the bottom figure shows the behavior of a safety maneuver. The trigger alternates between safety and finding the optimal location to trigger for inter-event time. Although there is no guarantee in an increase in inter-event time, we report that there are total of 215 trigger occurrences across the 2400 hours of orbital time, a reduction of 19.5 percent from the earlier simulation result.

\begin{figure}[thpb]
\label{fig:int_time}
\centering
\includegraphics[width=\linewidth,height=\textheight,keepaspectratio]{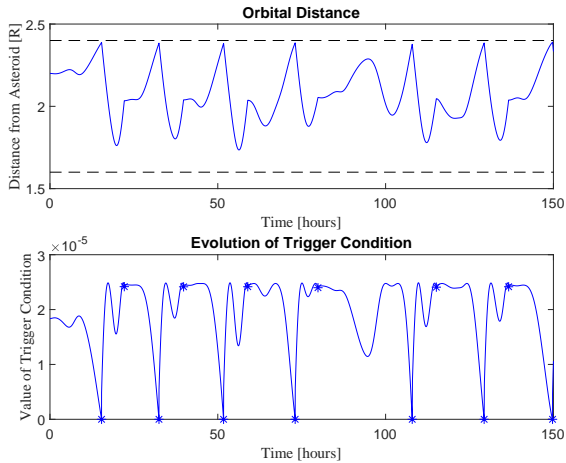}
\caption{Safety Maneuvers Simulation. [Top] Distance to the asteroid over 150 hours of satellite orbit. The plot shows the distance remains within the safe range at all time. [Bottom] Trigger condition over 150 hours of satellite orbit. Stars indicate the time at which the trigger occurs, showing the behavior of traditional spacecraft maneuvers. }
\end{figure}

\section{Event-Triggered Intermittent Safety Filter}
With the development of our event-triggered control for safety in impulsive systems, we can combine the different elements from our earlier results to develop the framework for intermittent safety filter. For intermittent systems, the dynamics is characterized by alternating periods where controller are on and off. We use this to describe when given safety filters are active and inactive.

We consider the application of safety filters with the intermittent dynamics given in~\eqref{sys:intermittent}. For our trigger framework, we use event-triggered control to determine when the filter needs to be back on in intermittent control system~\eqref{sys:intermittent}:
\begin{subequations}\label{trigger:filter_on}
\begin{align}
    t_{i+1}^\on &= \min\setdefB{t\geq t_{i}^\off}{\Xi_\on(x(t))\leq 0},\\
    \Xi_\on(x) &= \Lie_fh(x,k_\nom(x)) - \left\|\left. \frac{\partial h}{\partial x} \right|_{x}\right\| \bar d  +\alpha\big(h(x)\big).\label{eq:filter_off_cond}
\end{align}
\end{subequations}
This trigger relies on the same idea of monitoring the barrier condition and turning the filter back on when the condition gets violated. Indeed, in order to establish the MIET of the off duration, as we have studied in the impulsive control systems in Section~\ref{sec:ET-impulse}, we would require that $
\Xi_\on(x(t^\off_i))\geq c$ at time $t^\off_i$. To assure this is true, we use another trigger to determine when we can turn the filter off:
\begin{subequations}\label{trigger:filter_off}
\begin{align}
    t_{i}^\off &= \min\setdefB{t\geq t_{i}^\on}{\Xi_\off(x(t))\leq 0},\\
    \Xi_\off(x) &= \Xi_\on(x)-c.\label{eq:filter_off_cond}
\end{align}
\end{subequations}
The final key element in our framework is to guarantee that the above will occur. To this end, we will use the idea of increasing the value of barrier function~$h$ which will increases the value of the trigger condition~\eqref{eq:filter_off_cond}. Hence, we modify the constraint filter:
\begin{align}
k(x) = &\argmin_{u\in\real^m} \|u-k_\nom(x)\|^2\label{eq:filter_inc}\\
&\quad\text{s.t.}~~ \Lie_fh(x,u)- \left\|\left. \frac{\partial h}{\partial x} \right|_{x}\right\| \bar d\geq b,\nonumber
\end{align}
where $b>0$ is a positive constant. We will simply assume that the filter is feasible, and leave feasibility as a line of future research. In any case, even with the barrier function increasing, the trigger might still not occur because the value of $L_fh(x)$ may dominate $\alpha(h(x))$. Therefore, we assume $\alpha$ is large enough so that the nominal controller satisfy the barrier condition, at least for large value of $h$.

\begin{assumption}\longthmtitle{Nominal Safety}
\label{assump:nom_safe}
Given a nominal controller $k_\nom$, the function $\alpha$ is such that 
$$
\Lie_fh(x,k_\nom(x)) - \left\|\left. \frac{\partial h}{\partial x} \right|_{x}\right\| \bar d   \geq -\alpha\big(h(x)\big)+c
$$
for all $x\in\Cc$ such that $h(x)\geq\bar h$ for some positive $\bar h>0$.~\hfill$\bullet$
\end{assumption}
The assumption is related to the existence of a safety level (as described by $h$) where the nominal controller may operate without any filter. With this assumption, we assure that our safety promoting controller can drive the trajectories to such safe level, and therefore the off trigger will occur in finite time. The assumption in itself is not a strict one because a user usually gets to pick $\alpha$ and there always exists $\alpha$ large enough for the assumption to hold. Note however that the implication of choosing a large $\alpha$ lead to a less conservatism in safety because the trajectory is allowed to approach the boundary at a faster rate.

Now, we have all the elements for our intermittent safety filter framework. We are ready to give the following result.
\begin{theorem}\longthmtitle{Event-triggered Intermittent Safety Filter}
\label{thm:intermittent_filter}
Consider the intermittent nonlinear system~\eqref{sys:intermittent} with a nominal controller $k_\nom$ satisfying Assumption~\ref{assump:nom_safe} and a safety-filtered controller given by \eqref{eq:filter_inc}. Let the trigger designs~\eqref{trigger:filter_on} and~\eqref{trigger:filter_off} determine the time sequences $\{t_i^\on\}_{i\in\naturals}$ and  $\{t_i^\off\}_{i\in\naturals}$ iteratively.
Then there exists $t^\off_i$ for every $t^\on_i$.
In addition, under the same set of assumptions as in Proposition~\ref{prop:ET-impulse}, 
there exists a MIET for the off period, i.e., $\tau \leq t_{i+1}^\on - t_i^\off$ for all $i\in \naturals$. Consequently, $x(t)\in\Cc$ for all time if $x_0\in\Cc$. That is, the set $\Cc$ is safe. 
\end{theorem}
\begin{proof}
First, we prove the existence of $t_i^\off$. Due to the constraint in the filter~\eqref{eq:filter_inc}, we can deduce $\frac{dh}{dt} \geq b$
Therefore, $\bar h - h(x(t_i^\on))$ is reached in finite time $T \leq (\bar h - h(x(t_i^\on)))/b$. At which point, the trigger criterion~\eqref{trigger:filter_off} must already be satisfied.

The proof of MIET is as in the proof of Proposition~\ref{prop:ET-impulse}. With a MIET established, we can conclude that all maximal solutions are complete, i.e., exists for all time. Then, for time period $[t^\on_i,t^\off_{i})$, safety is guarantee due to the satisfaction of constraint in safety filter~\eqref{eq:filter_inc}. In addition, for time period $[t^\off_i,t^\on_{i+1})$, safety is guarantee due to the monitoring of the trigger condition~\eqref{eq:filter_off_cond}.  Thus, it can be determined iteratively that $x(t)\in\Cc$ at all time if $x_0\in\Cc$ using the barrier function~$h$, concluding the proof.
\end{proof}
Theorem \ref{thm:intermittent_filter} formalizes our event-triggered intermittent safety filter framework. We summarize how our framework maintain safety as follows. We no longer use the barrier condition to filter the nominal controller. Instead, the filter aims to promote safety by increasing the barrier function in order to maneuver into states where it is possible to turn the filter off. We only use barrier condition to monitor safety and when to filter. The trigger framework effectively add hysteresis to the system, allowing for a switching period between filtering and not filtering.

\section{Conclusion}
In this paper, we have proposed various trigger designs for the purpose of reducing control effort for safety objectives. We have developed trigger schemes for safeguarding controllers in impulsive control systems and for safety filters in nonlinear systems. One particular interesting idea explored is safety maneuver which switches between actively using control effort for safety and only monitoring safety. Our future work includes the application of our event-triggered intermittent safety filter on a robotic system with the goal of acoomplishing simultaneously a nominal task and  collision avoidance. In addition, we will analyze of the tradeoff between safety maneuver and progress towards nominal objective, particularly in the context the optimization-based controller with Lyapunov and barrier condition as constraints. Our hope is that safety maneuvers will allow us to make guarantee for satisfaction of nominal objectives.

\vspace{0.4cm}
\textbf{Acknowledgement.} The authors would like to thank JPL for their feedback and discussion on the control application to satellite problems. We especially thank Saptarshi Bandyopadhyay in particular for the suggestion of imitating orbit transfers and for providing the eighth-order harmonics gravity model used in our simulation results.

\bibliography{alias,PO,Main-Pio}
\bibliographystyle{ieeetr}
\end{document}